%% file: main.tex
\documentclass[12pt]{amsart}
\input{preambles}
\usepackage[margin=3cm]{geometry}
\usepackage {dashbox}
\usepackage{multicol}

\title{Undecidability of Linear Logics without Weakening}
\author{Jun Suzuki}
\address{Graduate School of Humanities and Human Sciences, 
  Hokkaido University}
\email{suzuki.jun.g2@elms.hokudai.ac.jp }
\author{Katsuhiko Sano}
  \address{Faculty of Humanities and Human Sciences, 
  Hokkaido University}
  \email{v-sano@let.hokudai.ac.jp}

\begin{document}
\begin{abstract}
The goal of this paper is to establish that it remains undecidable whether a sequent is provable in two systems in which a weakening rule for an exponential modality is completely omitted from classical propositional linear logic $\mathbf{CLL}$ introduced by Girard (1987), which is shown to be undecidable by Lincoln et al. (1992). We introduce two logical systems, $\mathbf{CLLR}$ and $\mathbf{CLLRR}$. The first system, $\mathbf{CLLR}$, is obtained by omitting the weakening rule for the exponential modality of $\mathbf{CLL}$. 
The system $\mathbf{CLLR}$ has been studied by several authors, including Melliès-Tabareau (2010), but its undecidability was unknown. This paper shows the undecidability of $\mathbf{CLLR}$ by reducing it to the undecidability of $\mathbf{CLL}$, where the units $\mathbf{1}$ and $\bot$ play a crucial role in simulating the weakening rule. We also omit these units from the syntax and inference rules of $\mathbf{CLLR}$ in order to define the second system, $\mathbf{CLLRR}$. The undecidability of $\mathbf{CLLRR}$ is established by showing that the system can simulate any two-counter machine proposed by Minsky (1961).
\end{abstract}
\maketitle

\section{Introduction}
The goal of this paper is to establish that it remains undecidable whether a sequent is provable in two systems in which a weakening rule for an exponential modality is completely omitted from classical propositional linear logic. 
Linear logic, introduced by Girard~\cite{Girard1987}, is a logical system in which the structural rules of weakening and contraction, representing discard and copy, respectively, are restricted. 
Because of this restriction, formulas cannot be freely copied and discarded, so that formulas in the antecedent of a sequent can be considered as resources. However, we can apply weakening and contraction to a formula in the antecedent prefixed with the exponential modality ``$\oc$''. Resources with ``$\oc$'' can be freely copied and discarded. What then if we restrict the structural rules of the exponential modality? This paper studies systems obtained by omitting the weakening rules of the exponential modality and shows their undecidability. 

Structural rules and decidability of proof systems have a subtle relationship. Table \ref{Strucural Rules and Decidability} summarizes whether linear logic systems with added or excluded structural rules are decidable. Classical propositional linear logic $\mathbf{CLL}$, in which weakening and contraction apply only to formulas with exponential modalities, is undecidable (Lincoln et al. \cite[Theorem 3.7]{Lincoln1992}). If one removes the restriction of weakening and contraction in $\mathbf{CLL}$, i.e., adds unlimited weakening and unlimited contraction to $\mathbf{CLL}$, then the system becomes equivalent to a propositional modal logic $\mathbf{S4}$, which is decidable. The system obtained by adding only unlimited weakening to $\mathbf{CLL}$ is also decidable (Kopylov \cite[Theorem 3]{Kopylov2001}). Similarly, $\mathbf{CLL}$ with unlimited contraction is decidable (Okada-Terui \cite[Corollary 1]{OkadaTerui1999}). What if one omits structural rules for exponential modalities? For example, $\mathbf{MALL}$, classical propositional linear logic without exponential modalities, which does have neither weakening nor contraction, is decidable (Lincoln et al. \cite[Theorem 2.2]{Lincoln1992}). However, non-commutative classical propositional linear logic, which we can regard as $\mathbf{CLL}$ without the third structural rule \emph{exchange}, is undecidable (Lincoln et al. \cite[Theorem 4.8]{Lincoln1992}). Non-commutative $\mathbf{CLL}$ is so complex that it is undecidable if one omits units and additive connectives. Furthermore, non-commutative propositional linear logic without weakening is also undecidable. This system is still undecidable if one omits units $\mathbf{1}$ and $\bot$ from its language. These undecidabilities can be deduced by a theorem and a corollary about non-commutative subexponential linear logic in Kanovich et al. \cite[Theorem 8, Corollary 14]{KanovichKuznetsovNigamScedrov2019}. Subexponential linear logic is explained below.

\begin{table}[htbp]
    \renewcommand{\arraystretch}{1.0}
    \caption{Structural Rules and Decidability}
    \label{Strucural Rules and Decidability}
    \begin{tabular}{|c|c|c|}
    \hline
    System & $+/-$ Structural rules & Decidable or Undecidable\\\hline
    $\mathbf{S4}$ & $\mathbf{CLL}$ $+$ \{Weak, Cont\} & \begin{tabular}{c} \textbf{Decidable}\\ cf. Blackburn et al. \cite[Corollary 6.8]{BlackburnRijkeVenema_2001} \end{tabular} \\\hline
    \textbf{CLLW} & $\mathbf{CLL}$ $+$ Weakening & \begin{tabular}{c}\textbf{Decidable}\\
    Kopylov~\cite[Theorem 3]{Kopylov2001}
    \end{tabular}\\\hline
    \textbf{CLLC} & $\mathbf{CLL}$ $+$ Contraction & \begin{tabular}{c}\textbf{Decidable}\\ Okada-Terui \cite[Corollary 1]{OkadaTerui1999}\end{tabular}\\\hline
    $\mathbf{CLL}$ & $\mathbf{CLL}$ & \begin{tabular}{c} \textbf{Undecidable}\\ Lincoln et al. \cite[Theorem 3.7]{Lincoln1992} \end{tabular}\\\hline
    \textbf{NCCLL} & $\mathbf{CLL}$ $-$ Exchange & \begin{tabular}{c}\textbf{Undecidable}\\ Lincoln et al. \cite[Theorem 4.8]{Lincoln1992} \end{tabular}\\\hline
    $\mathbf{CLLR}$ & $\mathbf{CLL}$ $-$ Weakening & \begin{tabular}{c}\textbf{Undecidable} \\ Theorem \ref{Main Theorem 1} of this paper \end{tabular} \\\hline
    $\mathbf{CLLRR}$ & \begin{tabular}{c}
         $\mathbf{CLL}$ $-$ Weakening\\
         Lang: $\mathbf{CLL} - \{\mathbf{1}, \bot\}$
    \end{tabular}  &\begin{tabular}{c} \textbf{Undecidable} \\ Theorem \ref{UndCLLRR} of this paper \end{tabular} \\\hline
      & $\mathbf{CLL} - \{\text{Weak, Exch}\}$ & \begin{tabular}{c} \textbf{Undecidable} \\ Kanovich et al. \cite[Theorem 8]{KanovichKuznetsovNigamScedrov2019} \end{tabular}\\\hline
     &
    \begin{tabular}{c}
         $\mathbf{CLL}$ $-$ \{Weak, Exch\}\\
         Language: $\mathbf{CLL} - \{\mathbf{1}, \bot\}$
    \end{tabular} & \begin{tabular}{c} \textbf{Undecidable} \\ Kanovich et al. \cite[Corollary 14]{KanovichKuznetsovNigamScedrov2019}\end{tabular}\\\hline
    \textbf{MALL} & Language: $\mathbf{CLL}$ $-$ \{$\oc, \wn$\} & \begin{tabular}{c} \textbf{Decidable} \\ Lincoln et al. \cite[Theorem 2.2]{Lincoln1992} \end{tabular}\\\hline
    \end{tabular}
    \centering
\end{table}
Then it is not obvious whether classical propositional linear logic without only weakening is decidable.
In this paper, we introduce two systems, $\mathbf{CLLR}$ and $\mathbf{CLLRR}$, and show that both of them are undecidable. The first system, $\mathbf{CLLR}$, is obtained by omitting weakening for the exponential modality of $\mathbf{CLL}$. 
This system was introduced by Melliès-Tabareau \cite{MelliesTabareau2010}\footnote{Melliès-Tabareau refer to Jacobs \cite{Jacobs1994} as previous research (\cite[p. 636]{MelliesTabareau2010}), but a sequent calculus that Jacobs introduced is different from that of Melliès-Tabareau and ours.} and studied for its game semantics and categorical interpretation, but its decidability has been still open. 
As we will see in Section \ref{Undecidability of CLLR}, since $\mathbf{CLLR}$ can simulate weakening, the undecidability of $\mathbf{CLLR}$ can be reduced to that of $\mathbf{CLL}$. 
In this simulation, the logical units $\mathbf{1}$ and $\bot$ play a crucial role. Therefore, we remove them from the syntax, along with the associated rules of $\mathbf{CLLR}$, in order to define a new system, $\mathbf{CLLRR}$. In this system, the simulation is impossible. Formally, our main results are as follows.
\begin{enumerate}
    \item The problem of whether a sequent is provable in $\mathbf{CLLR}$ is undecidable.
    \item The problem of whether a sequent is provable in $\mathbf{CLLRR}$ is undecidable.
\end{enumerate}

We consider linear logic systems in which contraction is restricted only to formulas in the antecedent prefixed with ``$\oc$'' (and dually in the succedent prefixed with ``$\wn$''), while weakening is not allowed at all. This system is capable of representing resources that can be freely copied but not discarded. A study of logics restricting structural rules of exponential modalities is initiated by Danos et al. \cite[Section 5]{DanosJoinetSchellinx1993}. The system introduced by them, called \emph{subexponential linear logic} after Nigam-Miller \cite{NigamMiller2009}, indexes exponential modalities and control whether structural rules can be applied. Our system, in which there is no modality that allows weakening at all, is a special case of subexponential linear logic. Danos et al. \cite[Proposition 5.2]{DanosJoinetSchellinx1993} proves cut-elimination theorem for subexponential linear logic, which yields cut-elimination of our systems, $\mathbf{CLLR}$ and $\mathbf{CLLRR}$. However, the paper~\cite{DanosJoinetSchellinx1993} did not address the question of whether the two systems are decidable. We can also use our arguments to demonstrate the undecidability of the intuitionistic variants of $\mathbf{CLLR}$ and $\mathbf{CLLRR}$.

This paper is structured as follows. Section \ref{Preliminaries} introduces syntax and proof systems of ordinary classical linear logic. Since we apply a semantic argument to prove the undecidability of $\mathbf{CLLRR}$, we also introduce phase semantics, a standard algebraic semantics of linear logic. In Section \ref{Undecidability of CLLR}, we establish the undecidability of $\mathbf{CLLR}$ (Theorem \ref{Main Theorem 1}) by reducing it to the undecidability of $\mathbf{CLL}$. In Section \ref{Undecidability of CLLRR}, we first introduce the notion of two-counter machine, a Turing complete computational model. Then we establish the undecidability of $\mathbf{CLLRR}$ (Theorem \ref{UndCLLRR}) by showing that it can simulate any two-counter machine with a semantic argument employed in Lafont \cite{Lafont1996}. Section \ref{comclusion} shows that the results can be extended to intuitionistic linear logic. Section \ref{comclusion} concludes the paper with future research directions.

\section{Preliminaries}\label{Preliminaries}
\subsection{Syntax and Proof System of Classical Propositional Linear Logic}
We define a syntax $\mathcal{L}$ of classical propositional linear logic as follows: 
\begin{align*}
    \mathcal{L} \ni A \Coloneqq &~p \mid {\bf 1} \mid \bot \mid \top \mid {\bf 0} \mid {\sim}A \mid A \otimes A \mid A \parr A \mid A  \with A \mid A \oplus A \mid A \multimap A \mid \oc A \mid \wn A,
\end{align*}
where $p$ is an element of countably infinite set $\sf{Prop}$ of propositional variables. Conjunction and disjunction are each separated into two types: multiplicative (context splitting) and additive (context sharing), as in the following table.
\begin{table}[h]
    \centering
    \begin{tabular}{|c|c|c|}
        \hline
        & conjunction & disjunction \\\hline
        \begin{tabular}{c}
            multiplicative \\
            (context splitting)
        \end{tabular} & \Large{$\otimes$} & \Large{$\parr$} \\\hline
        \begin{tabular}{c}
            additive \\
            (context sharing)
        \end{tabular} & \Large{$\with$} & \Large{$\oplus$} \\\hline
    \end{tabular}
\end{table}
The symbols $\mathbf{1}$, $\bot$, $\top$, $\mathbf{0}$ are the units of $\otimes$, $\parr$, $\with$, $\oplus$ respectively. We use $\Gamma$, $\Delta$, etc. to denote finite multisets of formulas. An expression of the form $\Gamma \Rightarrow \Delta$ is called a \emph{sequent} of $\mathcal{L}$. A \emph{sequent calculus} $\mathbf{CLL}$ of classical propositional linear logic is shown in Table \ref{table:CLL}, where all of the rules indicated by the dashed box will be eliminated in two steps later in this paper.

\begin{table}[htbp]
    \caption{Sequent Calculus of $\mathbf{CLL}$}
    \centering
    \renewcommand{\arraystretch}{2.0}
    \begin{tabular}{|cc|}
    \hline
        \begin{bprooftree}
  \AxiomC{}
  \RightLabel{\textbf{id}}
  \UnaryInfC{$A \Rightarrow A$}
  \end{bprooftree}
&
  \begin{bprooftree}
  \AxiomC{$\Gamma \Rightarrow \Delta, A$}
  \AxiomC{$A, \Gamma' \Rightarrow \Delta'$}
  \RightLabel{$Cut$}
  \BinaryInfC{$\Gamma, \Gamma' \Rightarrow \Delta, \Delta'$}
  \end{bprooftree} \\
  \multicolumn{2}{|c|}{
  \begin{bprooftree}
  \AxiomC{$A, \Gamma \Rightarrow \Delta$}
  \RightLabel{[${\sim}r$]}
  \UnaryInfC{$\Gamma \Rightarrow \Delta, {\sim}A$}
  \end{bprooftree}
  \begin{bprooftree}
  \AxiomC{$\Gamma \Rightarrow \Delta, A$}
  \RightLabel{[${\sim}l$]}
  \UnaryInfC{${\sim}A, \Gamma \Rightarrow \Delta$}
  \end{bprooftree}
  \dbox{
\begin{bprooftree}
  \AxiomC{$\Gamma \Rightarrow \Delta$}
  \RightLabel{[$\bot r$]}
  \UnaryInfC{$\Gamma \Rightarrow \Delta, \bot$}
\end{bprooftree}
\begin{bprooftree}
  \AxiomC{}
  \RightLabel{[$\bot l$]}
  \UnaryInfC{$\bot \Rightarrow$}
\end{bprooftree}}
}
\\
  \multicolumn{2}{|c|}{
  \dbox{
\begin{bprooftree}
  \AxiomC{}
  \RightLabel{[$\mathbf{1} r$]}
  \UnaryInfC{$\Rightarrow \mathbf{1}$}
\end{bprooftree}
\begin{bprooftree}
  \AxiomC{$\Gamma \Rightarrow \Delta$}
  \RightLabel{[$\mathbf{1} l$]}
  \UnaryInfC{$\mathbf{1}, \Gamma \Rightarrow \Delta$}
\end{bprooftree}
}
\begin{bprooftree}
  \AxiomC{}
  \RightLabel{[$\top r$]}
  \UnaryInfC{$\Gamma \Rightarrow \Delta, \top$}
\end{bprooftree}
\begin{bprooftree}
  \AxiomC{}
  \RightLabel{[$\mathbf{0} l$]}
  \UnaryInfC{$\mathbf{0}, \Gamma \Rightarrow \Delta$}
\end{bprooftree}}\\
  \begin{bprooftree}
  \AxiomC{$\Gamma \Rightarrow \Delta, A$}
  \AxiomC{$\Gamma' \Rightarrow \Delta', B$}
  \RightLabel{[$\otimes r$]}
  \BinaryInfC{$\Gamma, \Gamma' \Rightarrow \Delta, \Delta', A \otimes B$}
  \end{bprooftree}
&
  \begin{bprooftree}
  \AxiomC{$A, B, \Gamma \Rightarrow \Delta$}
  \RightLabel{[$\otimes l$]}
  \UnaryInfC{$A \otimes B, \Gamma \Rightarrow \Delta$}
  \end{bprooftree}\\

  \begin{bprooftree}
  \AxiomC{$\Gamma \Rightarrow \Delta, A$}
  \AxiomC{$\Gamma \Rightarrow \Delta, B$}
  \RightLabel{[$\with r$]}
  \BinaryInfC{$\Gamma \Rightarrow \Delta, A \with B$}
  \end{bprooftree}
&
  \begin{bprooftree}
  \AxiomC{$A_i, \Gamma \Rightarrow \Delta$}
  \RightLabel{[$\with l_i$] $(i = 0, 1)$}
  \UnaryInfC{$A_0 \with A_1, \Gamma \Rightarrow \Delta$}
  \end{bprooftree}\\
  
  \begin{bprooftree}
  \AxiomC{$\Gamma \Rightarrow \Delta, A, B$}
  \RightLabel{[$\parr r$]}
  \UnaryInfC{$\Gamma \Rightarrow \Delta, A \parr B$}
  \end{bprooftree}
&
  \begin{bprooftree}
  \AxiomC{$A, \Gamma \Rightarrow \Delta$}
  \AxiomC{$B, \Gamma' \Rightarrow \Delta'$}
  \RightLabel{[$\parr l$]}
  \BinaryInfC{$A \parr B, \Gamma, \Gamma' \Rightarrow \Delta, \Delta'$}
  \end{bprooftree}\\
  
  \begin{bprooftree}
  \AxiomC{$\Gamma \Rightarrow \Delta, A_i$}
  \RightLabel{[$\oplus r_i$] $(i = 0, 1)$}
  \UnaryInfC{$\Gamma \Rightarrow \Delta, A_0 \oplus A_1$}
  \end{bprooftree}
&
  \begin{bprooftree}
  \AxiomC{$A, \Gamma \Rightarrow \Delta$}
  \AxiomC{$B, \Gamma \Rightarrow \Delta$}
  \RightLabel{[$\oplus l$]}
  \BinaryInfC{$A \oplus B, \Gamma \Rightarrow \Delta$}
  \end{bprooftree}\\

  \begin{bprooftree}
  \AxiomC{$A, \Gamma \Rightarrow \Delta, B$}
  \RightLabel{[$\multimap r$]}
  \UnaryInfC{$\Gamma \Rightarrow \Delta, A \multimap B$}
  \end{bprooftree}
&
  \begin{bprooftree}
  \AxiomC{$\Gamma \Rightarrow \Delta, A$}
  \AxiomC{$B, \Gamma' \Rightarrow \Delta'$}
  \RightLabel{[$\multimap l$]}
  \BinaryInfC{$A \multimap B, \Gamma, \Gamma' \Rightarrow \Delta, \Delta'$}
  \end{bprooftree}\\
\dbox{
  \begin{bprooftree}
  \AxiomC{$\Gamma \Rightarrow \Delta$}
  \RightLabel{[$\oc \mathbf{W}$]}
  \UnaryInfC{$\oc A, \Gamma \Rightarrow \Delta$}
  \end{bprooftree}
  }
&
  \begin{bprooftree}
  \AxiomC{$\oc A, \oc A, \Gamma \Rightarrow \Delta$}
  \RightLabel{[$\oc \mathbf{C}$]}
  \UnaryInfC{$\oc A, \Gamma \Rightarrow \Delta$}
  \end{bprooftree}
\\
  \begin{bprooftree}
  \AxiomC{$\oc \Gamma \Rightarrow \wn\Delta, A$}
  \RightLabel{[$\oc r$]}
  \UnaryInfC{$\oc \Gamma \Rightarrow \wn \Delta, \oc A$}
  \end{bprooftree}
&
  \begin{bprooftree}
  \AxiomC{$A, \Gamma \Rightarrow \Delta$}
  \RightLabel{[$\oc l$]}
  \UnaryInfC{$\oc A, \Gamma \Rightarrow \Delta$}
  \end{bprooftree}
\\
 \dbox{
  \begin{bprooftree}
  \AxiomC{$\Gamma \Rightarrow \Delta$}
  \RightLabel{[$\wn \mathbf{W}$]}
  \UnaryInfC{$\Gamma \Rightarrow \Delta, \wn A$}
  \end{bprooftree}
  }
&
  \begin{bprooftree}
  \AxiomC{$\Gamma \Rightarrow \Delta, \wn A, \wn A$}
  \RightLabel{[$\wn \mathbf{C}$]}
  \UnaryInfC{$\Gamma \Rightarrow \Delta, \wn A$}
  \end{bprooftree}
\\
  \begin{bprooftree}
  \AxiomC{$\Gamma \Rightarrow \Delta, A$}
  \RightLabel{[$\wn r$]}
  \UnaryInfC{$\Gamma \Rightarrow \Delta, \wn A$}
  \end{bprooftree}
&
  \begin{bprooftree}
  \AxiomC{$A, \oc\Gamma \Rightarrow \wn\Delta$}
  \RightLabel{[$\wn l$]}
  \UnaryInfC{$\wn A,\oc\Gamma \Rightarrow \wn\Delta$}
  \end{bprooftree}
\rule[-15pt]{0pt}{30pt}
  \\
\hline
    \end{tabular}
    \label{table:CLL}
\end{table}
\subsection{Phase Semantics and Soundness}
Let us introduce \emph{phase semantics}, a standard algebraic semantics of linear logic. Phase semantics was originally proposed in Girard \cite[Section 1]{Girard1987}, but we adopt the definition of Girard \cite[Section 2.1.2]{Girard1995}. See Okada \cite[Section 3]{Okada1998a} for more detail.

A \emph{phase space} is a pair $(\mathcal{M}, \bot)$ where $\mathcal{M} = (|\mathcal{M}|, \cdot, 1)$ is a commutative monoid and $\bot$ be an arbitrary subset of the domain $|\mathcal{M}|$ of $\mathcal{M}$.
Let $X, Y \subseteq |{\mathcal M}|$ in what follows. We define $X \cdot Y \subseteq |\mathcal{M}|$ as follows:
    \[X \cdot Y = \{x \cdot y \mid \text{$x \in X$ and $y \in Y$}\}.\]
The operator $``\cdot"$ and parentheses may be omitted as $(X \cdot Y) \cdot Z = XYZ$. Furthermore, define ${\sim}X \subseteq |{\mathcal M}|$ for $X$ as follows:
\[
    {\sim}X = \{y \in |{\mathcal M}| \mid (\forall x \in X)[x \cdot y \in \bot]\}.
\]
Then, $X \subseteq Y$ implies ${\sim}Y \subseteq {\sim}X$. We call $X$ a \emph{fact} if ${\sim\sim}X = X$. A phase space has the following property.
\begin{proposition}[cf. {Okada \cite[Lemma 1]{Okada1998a}}]\label{closure}
    Let $((|{\mathcal M}|, \cdot, 1), \bot)$ be a phase space. The operation ${\sim\sim}$ satisfies the following properties for any $X, Y \subseteq |\mathcal{M}|$:
    \begin{enumerate}
        \item $X \subseteq {\sim\sim}X$,
        \item ${\sim\sim}({\sim\sim}X) \subseteq {\sim\sim}X$,
        \item if $X \subseteq Y$ then, ${\sim\sim}X \subseteq {\sim\sim}Y$,
        \item ${\sim\sim}X \cdot {\sim\sim}Y \subseteq {\sim\sim}(XY)$.
    \end{enumerate}
\end{proposition}

Then we define a \emph{phase model} $\mathcal{P}$ $=$ $(({\mathcal M}, \bot), v)$ as a pair of a phase space $(\mathcal{M}, \bot)$ and a function $v \colon {\sf Prop} \to \wp(|{\mathcal M}|)$ such that $v$ satisfies the following: for all $p \in {\sf Prop}$, $v(p)$ is a fact, i.e., $v(p) = {\sim\sim}v(p)$.

Let ${\mathcal I} = \{i \in {\sim\sim}\{1\} \mid i \cdot i = i\}$. For a phase model $\mathcal{P} = ((|{\mathcal M}|, \cdot, 1), \bot, v)$, we define an \emph{interpretation} $[\cdot]_\mathcal{P} \colon \mathcal{L} \to \wp(|{\mathcal M}|)$ of formulas inductively as follows (if it is clear from the context which model is considered, the subscript can be omitted), although only those of $\otimes$, $\with$, $\oplus$ $\bot$, $\multimap$ and $\oc$ are used in this paper:
\begin{itemize}
    \item $[p] = v(p),\; [\mathbf{1}] = {\sim\sim}\{1\}, \; [\bot] = \bot, \; [\top] = |\mathcal{M}|, \; [\mathbf{0}] = {\sim\sim}\emptyset,\; [{\sim}A] = {\sim}[A],$
    \item $[A \otimes B] = {\sim\sim}([A][B]), \quad [A \with B] = [A] \cap [B],$
    \item $[A \parr B] = {\sim}({\sim}[A] \cdot {\sim}[B]), \quad [A \oplus B] = {\sim\sim}([A] \cup [B]),$
    \item $[A \multimap B] = \{z \in |\mathcal{M}| \mid (\forall x \in [A])[x \cdot z \in [B]]\},$
    \item $[\oc A] = {\sim}{\sim}([A] \cap {\mathcal I}), \quad [\wn A] = {\sim}({\sim}[A] \cap \mathcal{I}).$
\end{itemize}
It is noted that all of these interpretations are facts. We say that a formula $A \in \mathcal{L}$ is \emph{true} in $\mathcal{P}$ if $1 \in [A]_\mathcal{P}$. Classical propositional linear logic $\mathbf{CLL}$ is sound for the phase semantics as follows, where we define abbreviations $\bigotimes \Gamma$ and $\bigparr \Gamma$ for a finite multiset $\Gamma$ inductively as follows: $\bigotimes \big(\Gamma \cup \{A\}\big) = \big(\bigotimes \Gamma\big) \otimes A$, $\bigotimes \{A\} = A$, $\bigotimes \emptyset = \mathbf{1}$, and $\bigparr \big(\Gamma \cup \{A\}\big) = \big(\bigparr \Gamma \big) \parr A$, $\bigparr \{Y\} = Y$, $\bigparr \emptyset = \bot$.

\begin{proposition}[{cf. Okada \cite[Theorem 1]{Okada1998a}}]\label{soundness}
    Let $\Gamma \Rightarrow \Delta$ be a sequent of $\mathcal{L}$. If $\Gamma \Rightarrow \Delta$ is provable in $\mathbf{CLL}$, then for any phase model $\mathcal{P}$, $[\bigotimes \Gamma]_\mathcal{P} \subseteq [\bigparr \Delta]_\mathcal{P}$. In particular, for all $A \in \mathcal{L}$, if $\Rightarrow A$ is provable in $\mathbf{CLL}$ then for any phase model $\mathcal{P}$,  $A$ is true in $\mathcal{P}$.
\end{proposition}
\noindent It is noted that the latter statement follows from the former by $1 \in [\mathbf{1}] = \bigotimes \emptyset$. 
The following is useful in Section \ref{Undecidability of CLLRR}.
\begin{lemma}\label{inclusion}
    Let $\mathcal{P}$ be a phase model. If a sequent $\Gamma, A_1,..., A_n \Rightarrow C ~(n \geq 1)$ of $\mathcal{L}$ is provable in $\mathbf{CLL}$ and $B$ is true in $\mathcal{P}$ for all $B \in \Gamma$, then $[A_1]_\mathcal{P}\cdots[A_n]_\mathcal{P} \subseteq [C]_\mathcal{P}$.
\end{lemma}

\begin{proof}
    Let $n \geq 1$ and suppose that $\Gamma, A_1,\ldots, A_n \Rightarrow C$ is provable in $\mathbf{CLL}$ and that $B$ is true in $\mathcal{P}$ for all $B \in \Gamma$. 
    By Proposition \ref{soundness}, 
    $[\bigotimes (\Gamma \cup \{A_1,\ldots, A_n\}) ]_\mathcal{P} \subseteq [C]_\mathcal{P}$ holds. 
    Since an interpretation of any formula is a fact, Proposition \ref{closure} implies that  $[\bigotimes \{B_1, ..., B_m\}]_\mathcal{P} = [B_1]_\mathcal{P}\cdots [B_m]_\mathcal{P}$ 
    for all formulas $B_1, \ldots, B_m$. 
    Since $1 \in [B]_{\mathcal{P}}$ for all $B \in \Gamma$, we get $[A_1]_\mathcal{P}\cdots[A_n]_\mathcal{P} \subseteq [C]_\mathcal{P}$. 
\end{proof}
\section{Undecidability of CLLR}\label{Undecidability of CLLR}
We introduce a system of classical propositional linear logic without weakening, which already occurs in Melliès-Tabareau \cite{MelliesTabareau2010}.
\begin{definition}
    Define $\mathbf{CLLR}$ as a system obtained by excluding the rules $[\oc \mathbf{W}]$ and $[\wn \mathbf{W}]$ from $\mathbf{CLL}$ of Table \ref{table:CLL}.
\end{definition}
By translating formulas with the units $\mathbf{1}$ and $\bot$, 
this section shows that $\mathbf{CLLR}$ can simulate $\mathbf{CLL}$, which is undecidable. As a result, we show the undecidability of $\mathbf{CLLR}$ by this simulation. For this purpose, we define two translations $t_{l}$ and $t_{r}$, noting that similar translations are also found in Liang-Miller \cite{LiangMiller2009}.
\begin{definition}
    We define translations $t_l, t_r \colon \mathcal{L} \to \mathcal{L}$ by simultaneous induction as follows:
    \begin{itemize}
        \item $t_l(p) = t_r(p) = p,\; t_l(\mathbf{1}) = t_r(\mathbf{1}) = \mathbf{1},\; t_l(\top) = t_r(\top) = \top, \\
        t_l(\bot) = t_r(\bot) = \bot,\; t_l(\mathbf{0}) = t_r(\mathbf{0}) = \mathbf{0},$
        \item $t_l({\sim}A) = {\sim}t_r(A), \; t_r({\sim}A) = {\sim}t_l(A),$
        \item $t_l(A \circ B) = t_l(A) \circ t_l(B), \; t_r(A \circ B) = t_r(A) \circ t_r(B), \; (\circ \in \{\otimes, \with, \parr, \oplus\}),$
        \item $t_l(A \multimap B) = t_r(A) \multimap t_l(B), \; t_r(A \multimap B) = t_l(A) \multimap t_r(B),$
        \item $t_l(\oc A) = \oc(t_l(A) \with \mathbf{1}), \; t_r(\oc A) = \oc t_r(A),$
        \item $t_l(\wn A) = \wn t_l(A), \; t_r(\wn A) = \wn(t_r(A) \oplus \bot).$
    \end{itemize}
    For a finite multiset $\Gamma$, we write $t_l[\Gamma] = \{t_l(A) \mid A \in \Gamma\}$. Similarly for  $t_r[\Gamma]$.
\end{definition}
\noindent For example, $t_l(\oc(c_i \multimap a)) = \oc((c_i \multimap a) \with \mathbf{1})$. A method to simulate the weakening of $\oc$ through a translation with $\mathbf{1}$ is already suggested by Danos et al. \cite[p. 17]{DanosJoinetSchellinx1993}. 
For our translation, however, we must also show the equivalence of the provability of a sequent and its translation.
\begin{restatable}{proposition}{simulation}\label{Prop for simulation}
    For all $A \in \mathcal{L}$, $t_l(A) \Rightarrow t_r(A)$ is provable in $\mathbf{CLLR}$. 
\end{restatable}
\begin{proof}
    By induction on $A$, we show that $t_l(A) \Rightarrow t_r(A)$ is provable in $\mathbf{CLLR}$.\\
    (\textbf{Base Case}) Let $A \equiv p$, $\mathbf{1}$, $\bot$, $\top$, or $\mathbf{0}$. Since $t_l(A) = t_r(A)$, it suffices to apply \textbf{id}.\\
    (\textbf{Inductive Step}) Due to space limitations, only three cases are shown below.
    \begin{enumerate}
        \item Let $A \equiv {\sim}B$. We show that ${\sim}t_r(B) \Rightarrow {\sim}t_l(B)$ is provable. By induction hypothesis, there is a proof $\mathcal{D}_0$ of $t_l(B) \Rightarrow t_r(B)$. Then we proceed as follows:
        \[
            \begin{bprooftree}
                \AxiomC{$\mathcal{D}_0$}
                \noLine
                \UnaryInfC{$t_l(B) \Rightarrow t_r(B)$}
                \RightLabel{[${\sim}r$]}
                \UnaryInfC{$\Rightarrow {\sim}t_l(B), t_r(B)$}
                \RightLabel{[${\sim}l$]}
                \UnaryInfC{${\sim}t_r(B) \Rightarrow {\sim}t_l(B),$}
            \end{bprooftree}
        \]
        \item Let $A \equiv B \multimap C$. We show that $t_r(B) \multimap t_l(C) \Rightarrow t_l(B) \multimap t_r(C)$ is provable. By induction hypothesis, let $\mathcal{D}_0$ be a proof of $t_l(B) \Rightarrow t_r(B)$ and $\mathcal{D}_{1}$ be a proof of $t_l(C) \Rightarrow t_r(C)$. Then we proceed as follows:
        \[
            \begin{bprooftree}
                \AxiomC{$\mathcal{D}_1$}
                \noLine
                \UnaryInfC{$t_l(B) \Rightarrow t_r(B)$}
                \AxiomC{$\mathcal{D}_2$}
                \noLine
                \UnaryInfC{$t_l(C) \Rightarrow t_r(C)$}
                \RightLabel{[$\multimap l$]}
                \BinaryInfC{$t_r(B) \multimap t_l(C), t_l(B) \Rightarrow t_r(C)$}
                \RightLabel{[$\multimap r$]}
                \UnaryInfC{$t_r(B) \multimap t_l(C) \Rightarrow t_l(B) \multimap t_r(C),$}
            \end{bprooftree}
        \]
        \item Let $A \equiv \oc B$. We show that $\oc (t_l(B) \with \mathbf{1}) \Rightarrow \oc t_r(B)$ is provable. By induction hypothesis, there is a proof $\mathcal{D}_0$ of $t_l(B) \Rightarrow t_r(B)$. Then we proceed as follows:
        \[
            \begin{bprooftree}
                \AxiomC{$\mathcal{D}_0$}
                \noLine
                \UnaryInfC{$t_l(B) \Rightarrow t_r(B)$}
                \RightLabel{[$\with l_0$]}
                \UnaryInfC{$t_l(B) \with \mathbf{1} \Rightarrow t_r(B)$}
                \RightLabel{[$\oc l$]}
                \UnaryInfC{$\oc(t_l(B) \with \mathbf{1}) \Rightarrow t_r(B)$}
                \RightLabel{[$\oc r$]}
                \UnaryInfC{$\oc(t_l(B) \with \mathbf{1}) \Rightarrow \oc t_r(B)$} 
            \end{bprooftree}
        \]
    \end{enumerate}
\end{proof}
\begin{lemma}\label{simulateCLLbyCLLR1}
    Let $\Gamma \Rightarrow \Delta$ be a sequent of $\mathcal{L}$. If $\Gamma \Rightarrow \Delta$ is provable in $\mathbf{CLL}$, then $t_l[\Gamma] \Rightarrow t_r[\Delta]$ is provable in $\mathbf{CLLR}$.
\end{lemma}
\begin{proof}
    By induction on $\mathcal{D}$ of $\mathbf{CLL}$, we show that if $\textrm{root}(\mathcal{D}) = \Gamma \Rightarrow \Delta$ then $t_l[\Gamma] \Rightarrow t_r[\Delta]$ is provable in $\mathbf{CLLR}$. \\
    (\textbf{Base Case}) Since the translation does not change the units, the only case to consider is that of the $\mathbf{id}$ axiom. Let $\mathcal{D}$ be the $\mathbf{id}$ axiom. We need to show that $t_l(A) \Rightarrow t_r(A)$ is provable in $\mathbf{CLLR}$, but it is clear by Proposition \ref{Prop for simulation}.\\
    (\textbf{Inductive Step}) A crucial case is where the last applied rule is $[!\mathbf{W}]$:
    \[
    \mathcal{D} \equiv
    \begin{bprooftree}
        \AxiomC{$\mathcal{D}_0$}
        \noLine
        \UnaryInfC{$\Gamma \Rightarrow \Delta$}
        \RightLabel{[$\oc \mathbf{W}$]}
        \UnaryInfC{$\oc A, \Gamma \Rightarrow \Delta.$}
    \end{bprooftree}
    \]
    Since $t_l(\oc A) = \oc(t_l(A) \with \mathbf{1})$, we need to show that $\oc(t_l(A) \with \mathbf{1}), t_l[\Gamma] \Rightarrow t_r[\Delta]$ has a $\mathbf{CLLR}$ proof. 
    By induction hypothesis, there exists a proof in $\mathbf{CLLR}$ $\mathcal{D}'_0$ of $t_l[\Gamma] \Rightarrow t_r[\Delta]$. Then we proceed as follows:
    \[
    \begin{bprooftree}
        \AxiomC{$\mathcal{D}'_0$}
        \noLine
        \UnaryInfC{$t_l[\Gamma] \Rightarrow t_r[\Delta]$}
        \RightLabel{[$\mathbf{1}l$]}
        \UnaryInfC{$\mathbf{1}, t_l[\Gamma] \Rightarrow t_r[\Delta]$}
        \RightLabel{[$\with l_0$]}
        \UnaryInfC{$t_l(A) \with \mathbf{1}, t_l[\Gamma] \Rightarrow t_r[\Delta]$}
        \RightLabel{[$\oc l$]}
        \UnaryInfC{$\oc(t_l(A) \with \mathbf{1}), t_l[\Gamma] \Rightarrow t_r[\Delta].$}
    \end{bprooftree}
    \]
\end{proof}
In what follows, we are going to establish the converse of Lemma \ref{simulateCLLbyCLLR1}. 
\begin{proposition}\label{equivalence}
    For all $A \in \mathcal{L}$, both $A \Rightarrow t_l(A)$ and $t_r(A) \Rightarrow A$ are provable in $\mathbf{CLL}$.
\end{proposition}
\begin{proof}
    We show the two statements by simultaneous induction on $A$. The crucial case is where $A \equiv \oc B$. Since $t_l(\oc B) = \oc (t_l(B) \with \textbf{1})$ and $t_r(\oc B) = \oc t_r(B)$, we need to show that both $\oc B \Rightarrow \oc (t_l(B) \with \textbf{1})$ and $\oc t_r(B) \Rightarrow \oc B$ are provable in  $\mathbf{CLL}$. By induction hypothesis, there are a proof $\mathcal{D}_0$ of $\mathbf{CLL}$ whose root is $B \Rightarrow t_l(B)$ and $\mathcal{D}_1$ whose root is $t_r(B) \Rightarrow B$. Then we proceed as follows:
    \[
    \begin{bprooftree}
        \AxiomC{$\mathcal{D}_0$}
        \noLine
        \UnaryInfC{$B \Rightarrow t_l(B)$}
        \RightLabel{[$\oc l$]}
        \UnaryInfC{$\oc B \Rightarrow t_l(B)$}
        \AxiomC{}
        \RightLabel{[$\mathbf{1}r$]}
        \UnaryInfC{$\Rightarrow \mathbf{1}$}
        \RightLabel{[$\oc\mathbf{W}$]}
        \UnaryInfC{$\oc B \Rightarrow \mathbf{1}$}
        \RightLabel{[$\with r$]}
        \BinaryInfC{$\oc B \Rightarrow t_l(B) \with \mathbf{1}$}
        \RightLabel{[$\oc r$]}
        \UnaryInfC{$\oc B \Rightarrow \oc(t_l(B) \with \mathbf{1}),$}
    \end{bprooftree}
    \begin{bprooftree}
        \AxiomC{$\mathcal{D}_1$}
        \noLine
        \UnaryInfC{$t_r(B) \Rightarrow B$}
        \RightLabel{[$\oc l$]}
        \UnaryInfC{$\oc t_r(B) \Rightarrow B$}
        \RightLabel{[$\oc r$]}
        \UnaryInfC{$\oc t_r(B) \Rightarrow \oc B.$}
    \end{bprooftree}
    \]
\end{proof}

\begin{lemma}\label{simulateCLLbyCLLR2}
    Let $\Gamma \Rightarrow \Delta$ be a sequent of $\mathcal{L}$. If $t_l[\Gamma] \Rightarrow t_r[\Delta]$ is provable in $\mathbf{CLLR}$, then $\Gamma \Rightarrow \Delta$ is provable in $\mathbf{CLL}$.
\end{lemma}

\begin{proof}
    Suppose that $t_l[\Gamma] \Rightarrow t_r[\Delta]$ has a proof $\mathcal{D}$ in $\mathbf{CLLR}$. Since the axiom and the rules of $\mathbf{CLLR}$ are also those of $\mathbf{CLL}$, $\mathcal{D}$ is also a proof in $\mathbf{CLL}$. We can construct a proof of $\Gamma \Rightarrow \Delta$ in $\mathbf{CLL}$ as follows. 
    By Proposition \ref{equivalence}, we apply the cut rule finitely many times to obtain a proof of $\Gamma \Rightarrow \Delta$ in $\mathbf{CLL}$.
\end{proof}

By Lemmas \ref{simulateCLLbyCLLR1} and \ref{simulateCLLbyCLLR2},  $\mathbf{CLLR}$ can simulate $\mathbf{CLL}$ in the following sense.
\begin{corollary}
    Let $\Gamma \Rightarrow \Delta$ be a sequent of $\mathcal{L}$. Then  $t_l[\Gamma] \Rightarrow t_r[\Delta]$ is provable in $\mathbf{CLLR}$ iff $\Gamma \Rightarrow \Delta$ is provable in $\mathbf{CLL}$.
\end{corollary}
Hence, we can show the undecidability of $\mathbf{CLLR}$ by reducing it to the undecidability of $\mathbf{CLL}$ (Lincoln et al. \cite[Theorem 11]{Lincoln1992}).
\begin{theorem}\label{Main Theorem 1}
    The problem of whether a sequent is provable in $\mathbf{CLLR}$ is undecidable.
\end{theorem}
\section{Undecidability of CLLRR}\label{Undecidability of CLLRR}
This section establishes the undecidability of $\mathbf{CLLRR}$.
\begin{definition} 
    The syntax $\mathcal{L}^-$ of $\mathbf{CLLRR}$ is obtained by omitting $\mathbf{1}$ and $\bot$ from $\mathcal{L}$. 
    Define a system $\mathbf{CLLRR}$ of the syntax $\mathcal{L}^{-}$ as the resulting system obtained from $\mathbf{CLLR}$ by excluding the rules of $\mathbf{1}$ and $\bot$ of Table \ref{table:CLL}. 
\end{definition}
Since our reduction argument via two translations $t_{l}$ and $t_{r}$ for $\mathbf{CLLR}$ does not work for $\mathbf{CLLRR}$ due to the absence of the units $\mathbf{1}$ and $\bot$, this section establishes the undecidability of $\mathbf{CLLRR}$ in terms of two-counter machines \cite{Minsky1961}.

\subsection{Two-counter Machine}
We employ a formulation from Lafont \cite{Lafont1996}. A \emph{two-counter machine} $M$ consists of a finite set $S$ of states, the terminal state $s_t \in S$ and a function
\[
    \tau \colon S\backslash \{s_t\} \to (\{+\} \times \{A, B\} \times S) \cup (\{-\} \times \{A, B\} \times S \times S).
\]
An element of $S \times \mathbb{N} \times \mathbb{N}$ is said to be an \emph{instantaneous description} (\emph{ID}), which means a state and values of the two {\em counters}. For a given ID $(s_j, p, q)$, $\tau(s_j)$ represents a program that commands a transition from one ID to the next, which can take one of the following four forms: $(+, A, s_k), (-, A, s_k, s_l), (+, B, s_k), (-, B, s_k, s_l)$. The symbols ``$+$'' and ``$-$'' refer to increment and decrement commands respectively. The symbols ``$A$'' and ``$B$'' indicate which of the first and second counters should be increased or decreased.

Transitions of IDs by programs are as follows:
\begin{itemize}
    \item if $\tau(s_j) = (+, A, s_k)$, then $(s_j, p, q) \leadsto (s_k, p+1, q)$,
    \item if $\tau(s_j) = (-, A, s_k, s_l)$,
    \begin{itemize}
        \item if $p > 0$, then $(s_j, p, q) \leadsto (s_k, p-1, q)$, and
        \item if $p = 0$, then  $(s_j, p, q) \leadsto (s_l, p, q)$,
    \end{itemize}
    \item if $\tau(s_j) = (+, B, s_k)$, then $(s_j, p, q) \leadsto (s_k, p, q+1)$,
    \item if $\tau(s_j) = (-, B, s_k, s_l)$,
    \begin{itemize}
        \item if $q > 0$, then $(s_j, p, q) \leadsto (s_k, p, q-1)$, and        \item if $q = 0$, then  $(s_j, p, q) \leadsto (s_l, p, q)$,
    \end{itemize}
\end{itemize}
Provided an ID $(s_i, p, q)$, $M$ computes sequentially, starting with the program $\tau(s_i)$ corresponding to state $s_i$. Since there is no program corresponding to state $s_t$, the computation terminates when it reaches state $s_t$. We call an \emph{accepted sequence} of $(s_i, p, q)$ in $M$ a finite sequence of IDs
\[
    (s_0, p_0, q_0), (s_1, p_1, q_1), \ldots, (s_n, p_n, q_n),
\]
such that for all $s_k \in S\backslash\{s_t\}$, $\tau(s_{k-1})$ is a program that makes $(s_{k-1}, p_{k-1}, q_{k-1})$ transition to $(s_k, p_k, q_k)$, and $(s_0, p_0, q_0) = (s_i, p, q)$, $(s_n, p_n, q_n) = (s_t, 0, 0)$. Since the computation is deterministic, an accepted sequence is unique if it exists. An ID $(s_i, p, q)$ is \textit{accepted} by $M$ if there is the accepted sequence of $(s_i, p, q)$ in $M$.

Minsky \cite[Theorem Ia]{Minsky1961} establishes the existence of a machine whose acceptability problem is undecidable\footnote{
Minsky defines a slightly different machine from ours so that, regardless of the values of the counters, an ID is accepted when the machine reaches the terminal state. However, even if an ID is accepted only when it becomes $(s_t, 0, 0)$ as in this paper, we can still show the undecidability since the following holds:
\begin{quote}
    for any two-counter machine $M = (S, s_t, \tau)$, there is $M' = (S', s'_t, \tau')$ such that the computation in $M$ of input $(s_i, p, q)$ ends with $(s_t, r, s)$ for some $r, s \in \mathbb{N}$ iff $(s_i, p, q)$ is accepted by $M'$.
\end{quote}
Indeed, the following $M'$ satisfies this equivalence:
\begin{itemize}
    \item $S' = (S \times \{0\}) \cup (\{s'_0, s'_1\} \times \{1\})$, 
    \item $s'_t = s_t$,
    \item $\tau'$ is defined by
    \begin{itemize}
        \item if $s_i \in S \times \{0\}$, $\tau'(s_i)$ is obtained by replacing $s_t$'s in $\tau(s_i)$ with $s'_0$,
        \item $\tau'(s'_0, 1) = (-, A, (s'_0, 1), (s'_1, 1))$, and $\tau'(s'_1, 1) = (-, B, (s'_1, 1), s_t)$.
    \end{itemize}
\end{itemize}

Proof of the equivalence is as follows:

($\Longrightarrow$) Assume the left side of the equivalence. There exists the computation sequence of $(s_i, p, q)$ in $M'$. By carefully examining the definition of $\tau'$, we can show the right side. The program $\tau'(s'_0, 1)$ runs until the value of the first counter becomes $0$, at which point the state transitions to $(s'_1, 1)$. Next, the program $\tau'(s'_1, 1)$ runs until the value of the second counter becomes $0$. When the values of both counters reach to $0$, the state transitions to $0$ and the ID is accepted by $M'$.

($\Longleftarrow$) The contraposition clearly holds.
}.

\begin{restatable}{fact}{Und}\label{Und2cm}
    There exists a two-counter machine $M$ such that the problem of whether an input is accepted by $M$ is undecidable.
\end{restatable}
\subsection{Proof of the Undecidability of CLLRR}
Our method to show the undecidability uses phase semantics and is based on Lafont \cite{Lafont1996}, which established the undecidability of second-order multiplicative additive linear logic. Using phase semantics allows us to avoid a combinatorial argument of translating proofs into computations with lots of case distinctions, as in Lincoln et al \cite{Lincoln1992}. Although Lafont's argument is clear, it is not immediately obvious whether it can be applied to $\mathbf{CLLRR}$. We need to carefully translate programs of a two-counter machine into formulas.

We define the formula that is a translation of programs of a two-counter machine. Given a finite set $S = \{s_t, s_1, \ldots, s_n\}$ of states, we stipulate that  $c_t, c_1, \ldots, c_n$ are propositional variables corresponding to $s_t, s_1, \ldots, s_n$, respectively. 
\begin{definition}
\label{dfn:translation}
    Let $M = (S, s_t, \tau)$ be a two-counter machine. Fix propositional variables $a, b, a', b'$. We write $\theta_M$ for the formula obtained by connecting with $\with$ the set of the following formulas corresponding to the programs of $M$ and further four others:
    \begin{enumerate}
        \item for $\tau(s_j) = (+, A, s_k)$: $c_j \multimap c_k \otimes a$,
        \item for $\tau(s_j) = (-, A, s_k, s_l)$: $c_j \otimes a \multimap c_k$ and $c_j \multimap c_l \oplus (a' \with c_t)$,
        \item for $\tau(s_j) = (+, B, s_k)$: $c_j \multimap c_k \otimes b$,
        \item for $\tau(s_j) = (-, B, s_k, s_l)$: $c_j \otimes b \multimap c_k$ and $c_j \multimap c_l \oplus (b' \with c_t)$,
        \item $a' \multimap a' \with c_t$, $(a' \with c_t) \otimes b \multimap a' \with c_t$, $b' \multimap b' \with c_t$, $(b' \with c_t) \otimes a \multimap b' \with c_t$.
    \end{enumerate}
\end{definition}
The propositional variables ``$a$'', ``$b$'' correspond to two counters, while $a'$, $b'$ are introduced to deal with conditional branches of decrement commands. The implication ``$\multimap$'' represents the transition of states and the incrementing or decrementing of the counters. The variable ``$a$'' or ``$b$'' to the right of the implication, say in $c_j \multimap c_k \otimes a$ or $c_j \multimap c_k \otimes b$, corresponds to incrementing of the first or second counter, while ``$a$'' or ``$b$'' to the left of the implication, say in $c_j \otimes a \multimap c_k$ or $c_j \otimes b \multimap c_k$,  corresponds to decrementing of the first or second counter. Four formulas of item (5) are required to make a correspondence between an acceptance and the axiom.

This translation is obtained by modifying those of Kanovich \cite{Kanovich1995} and Lafont \cite{Lafont1996} introduced to show the undecidability of $\mathbf{CLL}$. In their translations, they chose
\begin{center}
    $c_j \multimap c_l \oplus a'$,
    $c_j \multimap c_l \oplus b'$
    \end{center}
instead of the second formulas item (2) and item (4) respectively and chose
    \begin{center}    
    $a' \multimap c_t$, 
    $a' \otimes b \multimap a'$, 
    $b' \multimap c_t$, 
    $b' \otimes a \multimap b'$
    \end{center}
instead of item (5) of Definition \ref{dfn:translation}. Furthermore, while Kanovich and Lafont used a finite multiset of formulas, we use a single formula $\theta_M$ connected by $\with$. This corresponds to the situation where only a necessary program is extracted. This becomes important when expressing a system where resources cannot be discarded.
\begin{lemma}\label{Main Theorem 2}
    For any two-counter machine $M = (S, s_t, \tau)$ and any ID $(s_i, p, q)$, if $(s_i, p, q)$ is accepted by $M$, then the sequent $(\oc \theta_M)^{g(s_i)}, c_i, a^p, b^q \Rightarrow c_t$ is provable in $\mathbf{CLLRR}$, where $g(s_i) = 0$ if $s_i = s_t$, otherwise $g(s_i) = 1$.
\end{lemma}

Note that $g: S \to \{0,1\}$ is a characteristic function of $S \setminus \{s_{t}\}$ and $g(s_{t})$ = $0$ means the situation where no more resources can be discarded. This is not used in Kanovich \cite{Kanovich1995} and Lafont \cite{Lafont1996}.
\begin{proof}[Proof of Lemma \ref{Main Theorem 2}]
    Fix any two-counter machine $M = (S, s_t, \tau)$. We show the following by mathematical induction on $n$ (when we refer to an ``accepted sequence'', we mean an accepted sequence in $M$):
    \begin{quote}
        For all ID $(s_i, p, q)$, if there is the accepted sequence of $(s_i, p, q)$ of length $n$ then the sequent $(\oc\theta_M)^{g(s_i)}, c_i, a^p, b^q \Rightarrow c_t$ is provable in $\mathbf{CLLRR}$.
    \end{quote}
    (\textbf{Base Case}) Let $n = 0$. Fix any ID $(s_i, p, q)$. Since there exists only one accepted sequence of length $0$, which is $(s_t, 0, 0)$ and we $g(s_t) = 0$, it suffice to show that $c_t \Rightarrow c_t$ is provable in $\mathbf{CLLRR}$, but 
    this is obvious.\\
    (\textbf{Inductive Step}) Let $n > 0$. Fix any ID $(s_i, p, q)$. There are six cases based on the ID $(s_i, p, q)$ and program $\tau(s_i)$. Due to space limitations, only three cases are shown below. The other cases can be shown similarly.
    \begin{enumerate}
        \item Let $\tau(s_i) = (+, A, s_j)$. Suppose that there exists the accepted sequence of length $n$: $(s_i, p, q), (s_j, p+1, q), \ldots, (s_t, 0, 0)$. Since $g(s_i) = 1$, we need to show that $\oc\theta_M, c_i, a^p, b^q \Rightarrow c_t$ is provable in $\mathbf{CLLRR}$. By induction hypothesis, there is a proof $\mathcal{D}$ of the sequent $(\oc\theta_M)^{g(s_j)}, c_j, a^{p+1}, b^q \Rightarrow c_t$. Then we proceed as follows:
        \[
            \begin{bprooftree}
                \AxiomC{}
                \RightLabel{\textbf{id}}
                \UnaryInfC{$c_i \Rightarrow c_i$}
                \AxiomC{$\mathcal{D}$}
                \noLine
                \UnaryInfC{$(\oc\theta_M)^{g(s_j)}, c_j, a^{p+1}, b^q \Rightarrow c_t$}
                \RightLabel{[$\otimes l$]}
                \UnaryInfC{$(\oc\theta_M)^{g(s_j)}, c_j \otimes a, a^p, b^q \Rightarrow c_t$}
                \RightLabel{[$\multimap l$]}
                \BinaryInfC{$(\oc\theta_M)^{g(s_j)}, c_i \multimap c_j \otimes a, c_i, a^p, b^q \Rightarrow c_t$}
                \doubleLine
                \RightLabel{[$\with l]^*$}
                \UnaryInfC{$(\oc\theta_M)^{g(s_j)}, \theta_M, c_i, a^p, b^q \Rightarrow c_t$}
                \RightLabel{[$\oc l$]}
                \UnaryInfC{$(\oc\theta_M)^{g(s_j)}, \oc\theta_M, c_i, a^p, b^q \Rightarrow c_t$}
                \RightLabel{[$\oc\mathbf{C}$]}
                \UnaryInfC{$\oc\theta_M, c_i, a^p, b^q \Rightarrow c_t.$}
            \end{bprooftree}
        \]
        Notice that if $g(s_j) = 0$, then we do not need to apply the last $\mathbf{C}$ rule.
        \item Let $\tau(s_i) = (-, A, s_j, s_k)$ and $p > 0$. Suppose that there exists the accepted sequence of length $n$: $(s_i, p, q), (s_j, p-1, q),\ldots, (s_t, 0, 0)$. Since $g(s_i) = 1$, we need to show that $\oc\theta_M, c_i, a^p, b^q \Rightarrow c_t$ is provable in $\mathbf{CLLRR}$. By induction hypothesis, there is a proof $\mathcal{D}$ of the sequent $(\oc\theta_M)^{g(s_j)}, c_j, a^{p-1}, b^q \Rightarrow c_t$. We obtain the following proof of $\oc\theta_M, c_i, a^p, b^q \Rightarrow c_t$ from it:
        \[
            \begin{bprooftree}
                \AxiomC{}
                \RightLabel{\textbf{id}}
                \UnaryInfC{$c_i \Rightarrow c_i$}
                \AxiomC{}
                \RightLabel{\textbf{id}}
                \UnaryInfC{$a \Rightarrow a$}
                \RightLabel{[$\otimes r$]}
                \BinaryInfC{$c_i, a \Rightarrow c_i \otimes a$}
                \AxiomC{$\mathcal{D}$}
                \noLine
                \UnaryInfC{$(\oc\theta_M)^{g(s_j)}, c_j, a^{p-1}, b^q \Rightarrow c_t$}
                \RightLabel{[$\multimap l$]}
                \BinaryInfC{$(\oc\theta_M)^{g(s_j)}, c_i \otimes a \multimap c_j, c_i, a^p, b^q \Rightarrow c_t$}
                \doubleLine
                \RightLabel{[$\with l]^*$}
                \UnaryInfC{$(\oc\theta_M)^{g(s_j)}, \theta_M, c_i, a^p, b^q \Rightarrow c_t$}
                \RightLabel{[$\oc l$]}
                \UnaryInfC{$(\oc\theta_M)^{g(s_j)}, \oc\theta_M, c_i, a^p, b^q \Rightarrow c_t$}
                \RightLabel{[$\oc\mathbf{C}$]}
                \UnaryInfC{$(\oc\theta_M)^{g(s_j)}, c_i, a^p, b^q \Rightarrow c_t.$}
            \end{bprooftree}
        \]
        Notice that if $g(s_j) = 0$, then we do not need to apply the last $\oc\mathbf{C}$ rule.
        \item Let $\tau(s_i) = (-, A, s_j, s_k)$ and $p = 0$. Suppose that there exists the accepted sequence of length $n$: $(s_i, 0, q), (s_k, 0, q), \ldots, (s_t, 0, 0)$. Since $g(s_i) = 1$, we need to show that $\oc \theta_M, c_i, b^q \Rightarrow c_t$ is provable in $\mathbf{CLLRR}$. There exists the accepted sequence of the ID $(s_k, 0, q)$ of length $n-1$. By induction hypothesis, there is a proof $\mathcal{D}$ of the sequent $(\oc\theta_M)^{g(s_k)}, c_k, b^q \Rightarrow c_t$. Furthermore, it can be shown by mathematical induction on $q \in \mathbb{N}$ that there is a proof $\mathcal{E}_q$ of $(\oc\theta_M)^{g(s_k)}, a' \with c_t, b^q \Rightarrow c_t$. We show by mathematical induction on $q \in \mathbb{N}$ that there is a proof $\mathcal{E}_q$ of $(\oc\theta_M)^{g(s_k)}, a' \with c_t, b^q \Rightarrow c_t$. First, let $q = 0$. If $g(s_k) = 0$, we can make a proof $\mathcal{E}_0$ of $a' \with c_t \Rightarrow c_t$ as follows:
        \[
            \begin{bprooftree}
                \AxiomC{}
                \RightLabel{\textbf{id}}
                \UnaryInfC{$c_t \Rightarrow c_t$}
                \RightLabel{[$\with l_1$]}
                \UnaryInfC{$a' \with c_t \Rightarrow c_t.$}
            \end{bprooftree}
        \]
        If $g(s_k) = 1$, we can make a proof $\mathcal{E}_0$ of $\oc\theta_M, a' \with c_t \Rightarrow c_t$ as follows:
        \[
            \begin{bprooftree}
                \AxiomC{}
                \RightLabel{\textbf{id}}
                \UnaryInfC{$a' \Rightarrow a'$}
                \RightLabel{[$\with l_0$]}
                \UnaryInfC{$a' \with c_t \Rightarrow a'$}
                \AxiomC{}
                \RightLabel{\textbf{id}}
                \UnaryInfC{$c_t \Rightarrow c_t$}
                \RightLabel{[$\with l_1$]}
                \UnaryInfC{$a' \with c_t \Rightarrow c_t$}
                \RightLabel{[$\multimap l$]}
                \BinaryInfC{$a' \multimap a' \with c_t, a' \with c_t \Rightarrow c_t$}
                \RightLabel{[$\with l]^*$}
                \doubleLine
                \UnaryInfC{$\theta_M, a' \with c_t \Rightarrow c_t$}
                \RightLabel{[$\oc l$]}
                \UnaryInfC{$\oc\theta_M, a' \with c_t \Rightarrow c_t.$}
            \end{bprooftree}
        \]
        Next, let $q > 0$. By induction hypothesis, $(\oc\theta_M)^{g(s_k)}, a' \with c_t, b^{q-1} \Rightarrow c_t$ has a proof $\mathcal{E}_{q-1}$. We obtain the following proof $\mathcal{E}_q$ of $\oc\theta_M, a' \with c_t, b^q \Rightarrow c_t$ from it:
        \[
            \begin{bprooftree}
                \AxiomC{}
                \RightLabel{\textbf{id}}
                \UnaryInfC{$a' \with c_t \Rightarrow a' \with c_t$}
                \AxiomC{}
                \RightLabel{\textbf{id}}
                \UnaryInfC{$b \Rightarrow b$}
                \RightLabel{[$\otimes r$]}
                \BinaryInfC{$a' \with c_t, b \Rightarrow (a' \with c_t) \otimes b$}
                \AxiomC{$\mathcal{E}_{q-1}$}
                \noLine
                \UnaryInfC{$(\oc\theta_M)^{g(s_k)}, a' \with c_t, b^{q-1} \Rightarrow c_t$}
                \RightLabel{[$\multimap l$]}
                \BinaryInfC{$(\oc\theta_M)^{g(s_k)}, (a' \with c_t) \otimes b \multimap a' \with c_t, a' \with c_t, b^q \Rightarrow c_t$}
                \doubleLine
                \RightLabel{[$\with l]^*$}
                \UnaryInfC{$(\oc\theta_M)^{g(s_k)}, \theta_M, a' \with c_t, b^q \Rightarrow c_t$}
                \RightLabel{[$\oc l$]}
                \UnaryInfC{$(\oc\theta_M)^{g(s_k)}, \oc\theta_M, a' \with c_t, b^q \Rightarrow c_t$}
                \RightLabel{[$\oc \mathbf{C}$]}
                \UnaryInfC{$\oc\theta_M, a' \with c_t, b^q \Rightarrow c_t.$}
            \end{bprooftree}
        \]
        Notice that if $g(s_k) = 0$, then we do not need to apply the last $\oc\mathbf{C}$ rule. Combining $\mathcal{D}$ and $\mathcal{E}_q$, we obtain the following proof of $\oc\theta_M, c_i, b^q \Rightarrow c_t$:
        \[
            \begin{bprooftree}
                \AxiomC{}
                \RightLabel{\textbf{id}}
                \UnaryInfC{$c_i \Rightarrow c_i$}
                \AxiomC{$\mathcal{D}$}
                \noLine
                \UnaryInfC{$(\oc\theta_M)^{g(s_k)}, c_k, b^q \Rightarrow c_t$}
                \AxiomC{$\mathcal{E}_q$}
                \noLine
                \UnaryInfC{$(\oc\theta_M)^{g(s_k)}, a' \with c_t, b^q \Rightarrow c_t$}
                \RightLabel{[$\oplus l$]}
                \BinaryInfC{$(\oc\theta_M)^{g(s_k)}, c_k \oplus (a' \with c_t), b^q \Rightarrow c_t$}
                \RightLabel{[$\multimap l$]}
                \BinaryInfC{$(\oc\theta_M)^{g(s_k)}, c_i \multimap c_k \oplus (a' \with c_t), c_i, b^q \Rightarrow c_t$}
                \RightLabel{[$\with l]^*$}
                \doubleLine
                \UnaryInfC{$(\oc\theta_M)^{g(s_k)}, \theta_M, c_i, b^q \Rightarrow c_t$}
                \RightLabel{[$\oc l$]}
                \UnaryInfC{$(\oc\theta_M)^{g(s_k)}, \oc\theta_M, c_i, b^q \Rightarrow c_t$}
                \RightLabel{[$\oc \mathbf{C}$]}
                \UnaryInfC{$\oc\theta_M, c_i, b^q \Rightarrow c_t.$}
            \end{bprooftree}
        \]
        Notice that if $g(s_k) = 0$, then we do not need to apply the last $\oc\mathbf{C}$ rule.\qedhere
    \end{enumerate}
\end{proof}

To prove the converse of Lemma \ref{Main Theorem 2}, we introduce a special kind of phase model, which is the same one as in Lafont \cite{Lafont1996}. In the following, we write $a^2b$ for a multiset $\{a, a, b\}$.
\begin{definition}
Given a two-counter machine $M = (S, s_t, \tau)$, the phase model $\mathcal{P}_M = ((\mathcal{M}, \bot), v)$ derived from $M$ is defined as follows:
\begin{itemize}
    \item $|\mathcal{M}| = \{\Gamma \mid \text{$\Gamma$ is a finite multiset of elements of $\sf{Prop}$}\}$. The unit $1 = \emptyset$. The monoid operator $\cdot = \cup$ $($the union operation of multisets$)$.
    \item $\bot$ is defined by
    \begin{align*}
        \bot = &~\{c_ia^pb^q \mid \text{$(s_i, p, q)$ is accepted by $M$}\} \cup \{a'b^q \mid q \in \mathbb{N}\} \cup \{b'a^p \mid p \in \mathbb{N}\}.
    \end{align*}
    \item $v(p) = {\sim\sim}\{p\}$.
\end{itemize}
\end{definition}

It is clear by definition of $\mathcal{P}_M$ that $\mathcal{I} = \{1\}$ and that $v(c_t) = \bot$.

\begin{lemma}\label{Main Lemma 2}
    For any two-counter machine $M = (S, s_t, \tau)$ and any ID $(s_i, p, q)$, if the sequent $(\oc \theta_M)^{g(s_i)}, c_i, a^p, b^q \Rightarrow c_t$ is provable in $\mathbf{CLLRR}$, then $(s_i, p, q)$ is accepted by $M$.
\end{lemma}
\begin{proof}
    Suppose that the sequent $(\oc \theta_M)^{g(s_i)}, c_i, a^p, b^q \Rightarrow c_t$ is provable in $\mathbf{CLLRR}$, which implies the sequent is provable also in $\mathbf{CLL}$.
    By Lemma \ref{inclusion}, if $\oc\theta_M$ is true in the phase model $\mathcal{P}_M$, i.e., $1 = \emptyset \in [\oc\theta_M]$, then $[c_i][a]^p[b]^q \subseteq [c_t]$, which means that, by the definition of $\mathcal{P}_M$ and Proposition \ref{closure}, $c_ia^pb^q \in v(c_t) = \bot$. If $c_ia^pb^q \in \bot$, by the definition of $\bot$, the ID $(s_i, p, q)$ is accepted by $M$. 
    So in what follows, it suffices to show that $\oc\theta_M$ is true in $\mathcal{P}_M$. By the fact that $\oc \theta_M = {\sim\sim}(\theta_M \cap \mathcal{I})$ and $\mathcal{I} = \{1\}$, and by the definition of $\with$, it suffices to show that all the formulas connected by $\with$ when defining $\theta_M$ are true in $\mathcal{P}_M$. We only show five cases because the others can be shown in a similar way.
    \begin{enumerate}
        \item Let $c_j \multimap c_k \otimes a$ be a component of $\theta_M$. To show that $1 \in [c_j \multimap c_k \otimes a]$, it suffices to show that $c_j \in {\sim\sim}\{c_ka\}$, which implies that $c_j \in {\sim\sim}({\sim\sim}\{c_k\} \cdot {\sim\sim}\{a\})$ since $\{c_ka\} \subseteq {\sim\sim}({\sim\sim}\{c_k\} \cdot {\sim\sim}\{a\})$ by Proposition \ref{closure}. Fix any $x \in {\sim}\{c_ka\}$. We show that $c_jx \in \bot$. Since $x \in {\sim}\{c_ka\}$, we have $c_kax \in \bot$. By the definition of $\bot$, $x$ is of the form $a^pb^q$, and $c_ka^{p+1}b^q \in \bot$, which means that the ID $(s_k, p+1, q)$ is accepted by $M$. By the definition of $\theta_M$, $\tau(s_j) = (+, A, s_k)$, so $(s_j, p, q)$ is also accepted by $M$. Therefore, $c_jx = c_ja^pb^q \in \bot$.
        \item Let $c_j \otimes a \multimap c_k$ be a component of $\theta_M$. It suffices to show that $c_ja \in {\sim\sim}\{c_k\}$, which implies that ${\sim\sim}\{c_ja\} \subseteq {\sim\sim}\{c_k\}$ by Proposition \ref{closure}. Fix any $x \in {\sim}\{c_k\}$. We show that $c_jax \in \bot$. Since $x \in {\sim}\{c_k\}$, we have $c_kx \in \bot$. By the definition of $\bot$, $x$ is of the form $a^pb^q$, and $c_ka^pb^q \in \bot$, which means that the ID $(s_k, p, q)$ is accepted by $M$. By the definition of $\theta_M$, $\tau(s_j) = (-, A, s_k, s_l)$, so $(s_j, p+1, q)$ is also accepted by $M$. Therefore, $c_jax = c_ja^{p+1}b^q \in \bot$.
        \item Let $c_j \multimap c_l \oplus (a' \with c_t)$ be a component of $\theta_M$. It suffices to show that $c_j \in {\sim\sim}({\sim\sim}\{c_l\} \cup ({\sim\sim}\{a'\} \cap \bot))$. Fix any $x \in {\sim}({\sim\sim}\{c_l\} \cup ({\sim\sim}\{a'\} \cap \bot))$. We show that $c_jx \in \bot$. Since ${\sim}({\sim\sim}\{c_l\} \cup ({\sim\sim}\{a'\} \cap \bot)) = {\sim}\{c_l\} \cap {\sim}({\sim\sim}\{a'\} \cap \bot)$, we have $x \in {\sim}({\sim\sim}\{a'\} \cap \bot)$. Since $a' \in {\sim\sim}\{a'\} \cap \bot$, $a'x \in \bot$ and so $x$ must be of the form $b^n$ for some $n \in \mathbb{N}$. Since $x \in {\sim}\{c_l\}$, $c_lx = c_lb^n \in \bot$, which means that the ID $(s_l, 0, n)$ is accepted by $M$. By the definition of $\theta_M$, $\tau(s_j) = (-, A, s_k, s_l)$, so $(s_j, 0, n)$ is also accepted by $M$. Therefore, $c_jb^n = c_jx \in \bot$.
        \item Let $a' \multimap a' \with c_t$ be a component of $\theta_M$. It suffices to show that $a' \in {\sim\sim}\{a'\} \cap \bot$. But this is clear since $a' \in {\sim\sim}\{a'\}$ by Proposition \ref{closure} and since $a' \in \bot$ by the definition of $\bot$.
        \item Let $(a' \with c_t) \otimes b \multimap a' \with c_t$ be a component of $\theta_M$. We show that $${\sim\sim}(({\sim\sim}\{a'\} \cap \bot) \cdot {\sim\sim}\{b\}) \subseteq {\sim\sim}\{a'\} \cap \bot.$$
    Since ${\sim\sim}(({\sim\sim}\{a'\} \cap \bot) \cdot {\sim\sim}\{b\}) \subseteq {\sim\sim}({\sim\sim}\{a'\} \cdot {\sim\sim}\{b\}) \subseteq {\sim\sim\sim\sim}\{a'b\} = {\sim\sim}\{a'b\}$ by Proposition \ref{closure}, it suffices to show that $a'b \in {\sim\sim}\{a'\} \cap \bot$, which implies that ${\sim\sim}\{a'b\} \subseteq {\sim\sim}\{a'\} \cap \bot$. But $a'b \in {\sim\sim}\{a'\}$ is shown by the fact that ${\sim}\{a'\} = \{b^n \mid n \in \mathbb{N}\}$ and that for all $n \in \mathbb{N}$, $a'bb^n = a'b^{n+1} \in \bot$. Furthermore, $a'b \in \bot$ is clear by the definition of $\bot$. \qedhere
    \end{enumerate}
\end{proof}

Hence, the following holds by Lemma \ref{Main Theorem 2} and Lemma \ref{Main Lemma 2}.
\begin{corollary}\label{Main Corrollary}
    For any two-counter machine $M = (S, s_t, \tau)$ and any ID $(s_i, p, q)$, the sequent $(\oc \theta_M)^{g(s_i)}, c_i, a^p, b^q \Rightarrow c_t$ is provable in $\mathbf{CLLRR}$ iff $(s_i, p, q)$ is accepted by $M$.
\end{corollary}

Combining this with Fact \ref{Und2cm}, we get the undecidability.
\begin{theorem}\label{UndCLLRR}
    The problem of whether a sequent is provable in $\mathbf{CLLRR}$ is undecidable.
\end{theorem}

\section{Intuitionistic Case}\label{intuitionistic}
A syntax $\mathcal{L}_\mathbf{I}$ of intuitionistic propositional linear logic is defined by
\[
    \mathcal{L}_\mathbf{I} \ni A \Coloneqq p \mid \mathbf{1} \mid \top \mid \mathbf{0} \mid A \otimes A \mid A \with A \mid A \oplus A \mid A \multimap A \mid \oc A.
\]
A \emph{sequent} of $\mathcal{L}_\mathbf{I}$ is of the form $\Gamma \Rightarrow C$, where $\Gamma$ is a finite multiset of formulas of $\mathcal{L}_\mathbf{I}$ and $C \in \mathcal{L}_\mathbf{I}$. A \emph{sequent calculus} $\mathbf{ILL}$ of intuitionistic propositional linear logic is shown in Table \ref{table:ILL}, where all of the rules indicated by the dashed box will be eliminated in the next. We define systems without weakening similarly as the classical case. A system obtained by excluding the [$\mathbf{W}$] rule from $\mathbf{ILL}$ is written as $\mathbf{ILLR}$, and $\mathbf{ILLRR}$ is defined to be a system obtained by excluding the rules of the unit $\mathbf{1}$ from $\mathbf{ILLR}$, whose syntax $\mathcal{L}^{-}_{\mathbf{I}}$ is obtained by omitting $\mathbf{1}$ from $\mathcal{L}_\mathbf{I}$.
\begin{table}[htbp]
    \caption{Sequent Calculus of $\mathbf{ILL}$}
    \centering
    \renewcommand{\arraystretch}{2.3}
    \begin{tabular}{|cc|}
    \hline\begin{bprooftree}
  \AxiomC{}
  \RightLabel{\textbf{id}}
  \UnaryInfC{$A \Rightarrow A$}
  \end{bprooftree}
&
  \begin{bprooftree}
  \AxiomC{$\Gamma \Rightarrow A$}
  \AxiomC{$A, \Delta \Rightarrow B$}
  \RightLabel{$Cut$}
  \BinaryInfC{$\Gamma, \Delta \Rightarrow B$}
  \end{bprooftree} \\

  \multicolumn{2}{|c|}{
  \dbox{
  \begin{bprooftree}
  \AxiomC{}
  \RightLabel{[$\mathbf{1} r$]}
  \UnaryInfC{$\Rightarrow \mathbf{1}$}
  \end{bprooftree}
\begin{bprooftree}
  \AxiomC{$\Gamma \Rightarrow C$}
  \RightLabel{[$\mathbf{1} l$]}
  \UnaryInfC{$\mathbf{1}, \Gamma \Rightarrow C$}
\end{bprooftree}
}
\begin{bprooftree}
  \AxiomC{}
  \RightLabel{[$\top r$]}
  \UnaryInfC{$\Gamma \Rightarrow \top$}
\end{bprooftree}
\begin{bprooftree}
  \AxiomC{}
  \RightLabel{[$\mathbf{0} l$]}
  \UnaryInfC{$\mathbf{0}, \Gamma \Rightarrow C$}
\end{bprooftree}}
\\
  \begin{bprooftree}
  \AxiomC{$\Gamma \Rightarrow A$}
  \AxiomC{$\Delta \Rightarrow B$}
  \RightLabel{[$\otimes r$]}
  \BinaryInfC{$\Gamma, \Delta \Rightarrow A \otimes B$}
  \end{bprooftree}
&
  \begin{bprooftree}
  \AxiomC{$A, B, \Gamma \Rightarrow C$}
  \RightLabel{[$\otimes l$]}
  \UnaryInfC{$A \otimes B, \Gamma \Rightarrow C$}
  \end{bprooftree}\\

  \begin{bprooftree}
  \AxiomC{$\Gamma \Rightarrow A$}
  \AxiomC{$\Gamma \Rightarrow B$}
  \RightLabel{[$\with r$]}
  \BinaryInfC{$\Gamma \Rightarrow A \with B$}
  \end{bprooftree}
&
  \begin{bprooftree}
  \AxiomC{$A_i, \Gamma \Rightarrow C$}
  \RightLabel{[$\with l_i$] $(i \in \{0, 1\})$}
  \UnaryInfC{$A_0 \with A_1, \Gamma \Rightarrow C$}
  \end{bprooftree}\\

  \begin{bprooftree}
  \AxiomC{$\Gamma \Rightarrow A_i$}
  \RightLabel{[$\oplus r_i$] $(i \in \{0, 1\})$}
  \UnaryInfC{$\Gamma \Rightarrow A_0 \oplus A_1$}
  \end{bprooftree}
&
  \begin{bprooftree}
  \AxiomC{$A, \Gamma \Rightarrow C$}
  \AxiomC{$B, \Gamma \Rightarrow C$}
  \RightLabel{[$\oplus l$]}
  \BinaryInfC{$A \oplus B, \Gamma \Rightarrow C$}
  \end{bprooftree}\\

  \begin{bprooftree}
  \AxiomC{$A, \Gamma \Rightarrow B$}
  \RightLabel{[$\multimap r$]}
  \UnaryInfC{$\Gamma \Rightarrow A \multimap B$}
  \end{bprooftree}
&
  \begin{bprooftree}
  \AxiomC{$\Gamma \Rightarrow A$}
  \AxiomC{$B, \Delta \Rightarrow C$}
  \RightLabel{[$\multimap l$]}
  \BinaryInfC{$A \multimap B, \Gamma, \Delta \Rightarrow C$}
  \end{bprooftree}\\\multicolumn{2}{|c|}{
  \dbox{
  \begin{bprooftree}
  \AxiomC{$\Gamma \Rightarrow C$}
  \RightLabel{[$\mathbf{W}$]}
  \UnaryInfC{$\oc A, \Gamma \Rightarrow C$}
  \end{bprooftree}
  }
  \begin{bprooftree}
  \AxiomC{$\oc A, \oc A, \Gamma \Rightarrow C$}
  \RightLabel{[$\mathbf{C}$]}
  \UnaryInfC{$\oc A, \Gamma \Rightarrow C$}
  \end{bprooftree}
  \begin{bprooftree}
  \AxiomC{$\oc \Gamma \Rightarrow A$}
  \RightLabel{[$\oc r$]}
  \UnaryInfC{$\oc \Gamma \Rightarrow \oc A$}
  \end{bprooftree}
  \begin{bprooftree}
  \AxiomC{$A, \Gamma \Rightarrow C$}
  \RightLabel{[$\oc l$]}
  \UnaryInfC{$\oc A, \Gamma \Rightarrow C$}
  \end{bprooftree}}
\rule[-15pt]{0pt}{30pt}\\ \hline
    \end{tabular}
    \label{table:ILL}
\end{table}

Define a translation $t \colon \mathcal{L}_\mathbf{I} \to \mathcal{L}_\mathbf{I}$ by induction as follows:
\begin{itemize}
    \item $t(p) = p,\; t(\mathbf{1}) = \mathbf{1},\; t(\top) = \top, \; t(\mathbf{0}) = \mathbf{0},$
    \item $t(A \circ B) = t(A) \circ t(B) \; (\circ \in \{\otimes, \with, \oplus, \multimap\}),$
    \item $t(\oc A) = \oc(t(A) \with \mathbf{1}).$
\end{itemize}
We write $t[\Gamma] = \{t(A) \mid A \in \Gamma\}$. Then, the following holds.
\begin{lemma}
    Let $\Gamma \Rightarrow C$ be a sequent of $\mathcal{L}_\mathbf{I}$. Then, $t[\Gamma] \Rightarrow C$ is provable in $\mathbf{ILLR}$ iff $\Gamma \Rightarrow C$ is provable in $\mathbf{ILL}$.
\end{lemma}

By this lemma, we can show the undecidability of $\mathbf{ILLR}$ by reducing it to the undecidability of $\mathbf{ILL}$, which is a corollary of the undecidability of $\mathbf{CLL}$.
\begin{theorem}
    The problem of whether a sequent is provable in $\mathbf{ILLR}$ is undecidable.
\end{theorem}

Our argument for the undecidability of $\mathbf{ILLRR}$ is contained in our proof of Theorem \ref{UndCLLRR}. Because our proof of Lemma \ref{Main Theorem 2} uses only the rules of $\mathbf{ILLRR}$, we can replace $\mathbf{CLLRR}$ with $\mathbf{ILLRR}$ in that lemma. Furthermore, since any sequent provable in $\mathbf{ILLRR}$ is provable also in $\mathbf{CLL}$, we can replace $\mathbf{CLLRR}$ with $\mathbf{ILLRR}$ also in Lemma \ref{Main Lemma 2}. Hence, we can also get the undecidability of $\mathbf{ILLRR}$.

\begin{theorem}
    The problem of whether a sequent is provable in $\mathbf{ILLRR}$ is undecidable.
\end{theorem}

\section{Conclusion and Further Direction}\label{comclusion}
This paper establishes the undecidability of two systems $\mathbf{CLLR}$ and $\mathbf{CLLRR}$ both of which can be regarded as the classical propositional linear logic without weakening.

We have two further directions to do in the future. Firstly, we may also study a linear logic system that retains weakening but omits a contraction rule. Since contraction is a main source of difficulty in a proof-search, we conjecture that linear logics without contraction are decidable. Secondly, we can study phase semantics for linear logics with restricted structural rules, including subexponential linear logic~\cite{DanosJoinetSchellinx1993}.
To prove the undecidability of $\mathbf{CLLRR}$, it is sufficient for us to use the soundness theorem for $\mathbf{CLL}$ in terms of phase semantics, and so we do not need to provide a sound and complete phase semantics for linear logics {\em without weakening}. However, we have a prospect of being able to make such a semantics by modifying $\mathcal{I}$ in the definition of the phase semantics. With respect to such a phase semantics, it would be interesting to investigate how to establish semantic arguments for the cut-elimination of linear logic without structural rules.

\bibliographystyle{splncs04}
\bibliography{Undecidability}
\end{document}

%% file: preambles.tex
\usepackage{graphicx} % Required for inserting images
\usepackage{enumitem}
\usepackage{amsmath}
\usepackage{amscd}
\usepackage{amsfonts}
\usepackage{ascmac}
\usepackage{cmll}
\usepackage{bussproofs}
\usepackage{amssymb}
\usepackage{mathtools}
\usepackage{latexsym}
\usepackage{here}
\usepackage{color}
\usepackage{url}
\newcommand\utimes{\mathbin{\ooalign{$\cup$\cr%
   \hfil\raise0.42ex\hbox{$\scriptscriptstyle\times$}\hfil\cr}}}
\usepackage{stmaryrd}
\usepackage{thm-restate}
\usepackage{prftree}
\theoremstyle{plain}
\theoremstyle{definition}
\newtheorem{proposition}{Proposition}
\newtheorem{theorem}[proposition]{Theorem}
\newtheorem{corollary}[proposition]{Corollary}

\newtheorem{lemma}[proposition]{Lemma}

\theoremstyle{definition}
\newtheorem{definition}[proposition]{Definition}

\newenvironment{bprooftree}
  {\leavevmode\hbox\bgroup}
  {\DisplayProof\egroup}
\makeatletter